\documentclass[12pt]{article}
\usepackage{dsfont}
\usepackage{amsmath}
\usepackage{amssymb}
\usepackage{amsfonts}
\usepackage{graphicx}
\usepackage{subfigure}
\usepackage{color}
\usepackage{amsmath}
\usepackage{enumerate}
\usepackage{amsthm}
\usepackage[all,poly]{xy}
\usepackage{multirow}
\usepackage{algorithm}
\usepackage{algorithmic}
\usepackage[noadjust]{cite}
\usepackage{epic,eepic,pstricks}
\usepackage[toc,page]{appendix}
\usepackage{natbib}
\usepackage{lipsum}
\usepackage{wasysym}

\newcommand{\bx}{\boldsymbol{x}}

\newcommand{\by}{\boldsymbol{y}}
\newcommand{\bu}{\boldsymbol{u}}
\newcommand{\bw}{\boldsymbol{w}}

\newcommand{\argmax}{\operatornamewithlimits{argmax}}

\newcommand{\prox}{\operatorname{prox}}

\theoremstyle{definition}
\newtheorem{defn}{Definition}[section]

\newtheorem{theo}{Theorem}[section]
\newtheorem{coro}{Corollary}[section]

\title{\bf Proximal Markov chain Monte Carlo algorithms}
\author{Marcelo Pereyra \hspace{.2cm}\\
University of Bristol, Department of Mathematics,\\ University Walk, Bristol, BS8 1TW, UK%\\ marcelo.pereyra@bristol.ac.uk, +44 (0)117 928-9146
}

\begin{document}

\def\spacingset#1{\renewcommand{\baselinestretch}%
{#1}\small\normalsize} \spacingset{1}

\maketitle

\begin{abstract}
This paper presents a new Metropolis-adjusted Langevin algorithm (MALA) that uses convex analysis to simulate efficiently from high-dimensional densities that are log-concave, a class of probability distributions that is widely used in modern high-dimensional statistics and data analysis. The method is based on a new first-order approximation for Langevin diffusions that exploits log-concavity to construct Markov chains with favourable convergence properties. This approximation is closely related to Moreau--Yoshida regularisations for convex functions and uses proximity mappings instead of gradient mappings to approximate the continuous-time process. The proposed method complements existing MALA methods in two ways. First, the method is shown to have very robust stability properties and to converge geometrically for many target densities for which other MALA are not geometric, or only if the step size is sufficiently small. Second, the method can be applied to high-dimensional target densities that are not continuously differentiable, a class of distributions that is increasingly used in image processing and machine learning and that is beyond the scope of existing MALA and HMC algorithms. To use this method it is necessary to compute or to approximate efficiently the proximity mappings of the logarithm of the target density. For several popular models, including many Bayesian models used in modern signal and image processing and machine learning, this can be achieved with convex optimisation algorithms and with approximations based on proximal splitting techniques, which can be implemented in parallel. The proposed method is demonstrated on two challenging high-dimensional and non-differentiable models related to image resolution enhancement and low-rank matrix estimation that are not well addressed by existing MCMC methodology.
\end{abstract}

\spacingset{1.45}
\section{Introduction}
\label{sec:intro}
With ever-increasing computational resources Monte Carlo sampling methods have become fundamental to modern statistical science and many of the disciplines it underpins. In particular, Markov chain Monte Carlo (MCMC) algorithms have emerged as a flexible and general purpose methodology that is now routinely applied in diverse areas ranging from statistical signal processing and machine learning to biology and social sciences. Monte Carlo sampling in high dimensions is generally challenging, especially in cases where standard techniques such as Gibbs sampling are not possible or ineffective. The most effective general purpose Monte Carlo methods for high-dimensional models are arguably the Metropolis-adjusted Langevin algorithms (MALA) \citep[p.371]{Robert} and Hamiltonian Monte Carlo (HMC) \citep{NealHMC}, two classes of MCMC methods that use gradient mappings to capture local properties of the target density and explore the parameter space efficiently. 

Advanced versions of MALA and HMC use other elements of differential calculus to achieve higher efficiency. For example, \citet{Minka2002} and \citet{Zhang2011} use Hessian matrices of the target density to capture higher-order information related to scale and correlation structure. Similarly, \citet{Girolami2011} use differential geometry to lift these methods from Euclidean spaces to Riemannian manifolds where the target density is isotropic. In this paper we move away from differential calculus and explore the potential of convex analysis for MCMC sampling from distributions that are log-concave.

Log-concave distributions, also known as ``convex models'' outside the statistical literature, are widely used in high-dimensional statistics and data analysis and, among other things, play a central role in revolutionary techniques such as compressive sensing and image super-resolution (see \citet{Candes2008,Chandrasekaran2012,Candes2009c} for examples in machine learning, signal and image processing, and high-dimensional statistics). Performing inference in these models is a challenging problem that currently receives a lot of attention. A major breakthrough on this topic has been the adoption of convex analysis in high-dimensional optimisation, which led to the development of the so-called ``proximal algorithms'' that use proximity mappings of concave functions, instead of gradient mappings, to construct fixed point schemes and compute function maxima (see \citet{Combettes2011} and \citet{BoydBook} for two recent tutorials on this topic). These algorithms are now routinely used to find the maximisers of posterior distributions that are log-concave and often non-smooth and very high-high-dimensionaldimensional \citep{Wainwright2012,ChandrasekaranPNAS2013, Figueiredo2011, Candes2009c, Candes2011, Chandrasekaran2011,Pesquet2012a}.

In this paper we use convex analysis and proximal techniques to construct a new Langevin MCMC method for high-dimensional distributions that are log-concave and possibly not continuously differentiable. Our experiments show that the method is potentially useful for performing Bayesian inference in many models related to signal and image processing that are not well addressed by existing MCMC methodology, for example, non-differentiable models with synthesis and analysis Laplace priors, priors related to total-variation, nuclear and elastic-net norms or with constraints to convex sets, such as norm balls and the positive semidefinite cone.

The remainder of the paper is structured as follows: Section \ref{ProbStat} specifies the class of distributions considered, defines some elements of convex analysis which are essential for our methods, and briefly recalls the unadjusted Langevin algorithm (ULA) and its Metropolised version MALA. In Section \ref{ss:pula} we present a proximal ULA for log-concave distributions and study its geometric convergence properties. Following on from this, Section \ref{ss:pmala} presents a proximal MALA which inherits the favourable convergence properties of the  unadjusted algorithm while guaranteeing convergence to the desired target density. Section \ref{sec:experiments} demonstrates the proposed methodology on two challenging high-dimensional applications related to image resolution enhancement and low-rank matrix estimation. Conclusions and potential extensions are finally discussed in Section \ref{sec:conclusion}. A MATLAB implementation of the proposed methods is available at \url{http://www.maths.bris.ac.uk/~mp12320/code/ProxMCMC.zip}.

\section{Definitions and notations}\label{ProbStat}
\subsection{Convex analysis}\label{ss:CA}
Let $\bx \in \mathbb{R}^n$ and let $\pi(d\bx)$ be a probability distribution which admits a density $\pi(\bx)$ with respect to the usual $n$-dimensional Lebesgue measure. We consider the problem of simulating from target densities of the form
\begin{equation}\label{pi}
 \pi(\bx) = \exp{\{g(\bx)\}}/\kappa,
\end{equation}
where $g: \mathbb{R}^n \rightarrow[0,\infty)$ is a concave upper semicontinuous function satisfying $\lim_{\|\bx\| \rightarrow \infty} g(\bx) = -\infty$. It is assumed that $g(\bx)$ can be evaluated point-wise and that the normalising constant $\kappa$ may be unknown. Although not denoted explicitly, $g$ may depend on the value of an observation vector, for instance in Bayesian inference problems. The methods presented in this paper will require $g$ to have a \emph{proximity mapping} that is inexpensive to evaluate or to approximate. %This is the case of many important statistical models used in modern statistical signal processing, image processing and machine learning (see [....] for examples in image restoration, compressive sensing, matrix recovery, etc...).

\begin{defn}{\textbf{Proximity mappings.}}
The $\lambda$-proximity mapping or proximal operator of a concave function $g$ is defined for any $\lambda > 0$ as \citep{Moreau1962}
\begin{equation}\label{proxMap}
\prox^\lambda_{g}(\bx) = \argmax_{\bu \in \mathbb{R}^n}\, g(\bu) -\|\bu - \bx\|^2/2\lambda.
\end{equation}
In order to gain intuition about this mapping it is useful to analyse its behaviour when the regularisation parameter $\lambda \in \mathbb{R}^+$ is either very small or very large. In the limit $\lambda \rightarrow \infty$, the quadratic penalty term vanishes and \eqref{proxMap} maps all points to the set of maximisers of $g$. In the opposite limit $\lambda \rightarrow 0$, the quadratic penalty dominates \eqref{proxMap} and the proximity mapping coincides with the identity operator, i.e., $\prox^\lambda_{g}(\bx) = \bx$. For finite values of $\lambda$, $\prox^\lambda_{g}(\bx)$ behaves similarly to a gradient mapping and moves points in the direction of the maximisers of $g$. Indeed, proximity mappings share many important properties with gradient mappings that are useful for devising fixed point methods, such as being firmly non-expansive, i.e., $\|\prox^\lambda_{g}(\bx)-\prox^\lambda_{g}(\by)\|^2 \leq (\bx-\by)^T\{\prox^\lambda_{g}(\bx)-\prox^\lambda_{g}(\by)\}, \forall \bx,\by \in \mathbb{R}^n$ \citep[Ch. 12]{CombettesBook}, and having the set of maximisers of $g$ as fixed points. These mappings were originally studied by \citet{Moreau1962}, \citet{Martinet} and \citet{Rockafellar} several decades ago. They have recently regained very significant attention in the convex optimisation community because of their capacity to move efficiently in high-dimensional and possibly non-differentiable scenarios, and are now used extensively in the proximal optimisation algorithms that underpin modern high-dimensional statistics, signal and image processing, and machine learning \citep{Combettes2011,Wainwright2012,ChandrasekaranPNAS2013,BoydBook}. Section \ref{PPMCMC} shows that proximity mappings are not only useful for optimisation, they also hold great potential for stochastic simulation.
\end{defn}

\begin{defn}{\textbf{Moreau approximations.}}
For any $\lambda > 0$, define the $\lambda$-Moreau approximation of $\pi$ as the following density
\begin{equation}\label{MoreauApprox}
\pi_\lambda(\bx) = \sup_{\bu \in \mathbb{R}^n} \, \pi(\bu) \exp{\left(-\|\bu - \bx\|^2/2\lambda\right)}/\kappa^\prime,
\end{equation}
with normalising constant $\kappa^\prime \in \mathbb{R}^+$. Moreau approximations \eqref{MoreauApprox} are closely related to Moreau--Yoshida envelope functions from convex analysis \citep{CombettesBook}. Precisely, $\log \pi_\lambda(\bx)$ is equal to the $\lambda$-Moreau-Yoshida envelope of $\log\pi(\bx)$ up to the additive constant $\log\kappa^\prime$. Note that $\pi_\lambda(\bx)$ can be efficiently evaluated (up to a constant) by using $\prox^\lambda_g(\bx)$, i.e., $\pi_\lambda(\bx) \propto  \exp{\left[g\{\prox^\lambda_{g}(\bx)\}\right]} \exp{\{-\|\prox^\lambda_{g}(\bx) - \bx\|^2/2\lambda\}}$.
\end{defn}

\begin{defn}{\textbf{Class of distributions $\mathcal{E}(\beta,\gamma)$}}
We say that $\pi$ belongs to the one-dimensional class of distributions with exponential tails $\mathcal{E}(\beta,\gamma)$ if for some $u$, and some constants $\gamma > 0$ and $\beta > 0$, $\pi$ takes the form 
\begin{equation}\label{classE}
\pi(x) \propto \exp{\left(-\gamma |x|^\beta\right)}, \quad |x| > u.
\end{equation}
\end{defn}

Moreau approximations have several properties that will be useful for constructing algorithms to simulate from $\pi$.

\noindent 1. \emph{Convergence to $\pi$}: The approximation $\pi_\lambda(\bx)$ converges point-wise to $\pi(\bx)$ as $\lambda \rightarrow 0$.\\
\noindent 2. \emph{Differentiability}: $\pi_\lambda(\bx)$ is continuously differentiable even if $\pi$ is not, and its log-gradient is $\nabla \log\pi_\lambda(\bx) = \{\prox^\lambda_{g}(\bx)-\bx\}/\lambda$.\\
\noindent 3. \emph{Subdifferential}: The point $\{\prox^\lambda_{g}(\bx)-\bx\}/\lambda$ belongs to the subdifferential\footnote{A vector $\bu \in \mathbb{R}^n$ is a subgradient of the concave function $g$ at the point $\bx_0 \in \mathbb{R}^n$ if $g(\bx) \leq g(\bx_0) + (\bx-\bx_0)^T\bu$ for all $\bx \in \mathbb{R}^n$. The set $\partial g(\bx_0)$ of all such subgradients  is called the subdifferential set of $g$ at the point $\bx_0$.} set of $\log\pi$ at $\prox^\lambda_{g}(\bx)$, i.e., $\{\prox^\lambda_{g}(\bx)-\bx\}/\lambda \in \partial\log\pi\{\prox^\lambda_{g}(\bx)\}$ \cite[Ch. 16]{CombettesBook}. In addition, if $\log\pi$ is differentiable at $\prox^\lambda_{g}(\bx)$ then its subdifferential collapses to a single point, i.e., $\{\prox^\lambda_{g}(\bx)-\bx\}/\lambda =  \nabla\log\pi\{\prox^\lambda_{g}(\bx)\}$.\\
\noindent 4. \emph{Maximizers}: The set of maximizers of $\pi_\lambda$ is equal to that of $\pi$. Also, because $\pi_\lambda$ is continuously differentiable, $\nabla \log\pi_\lambda(\bx^*) = 0 $ implies that $\bx^*$ is a maximizer of $\pi$.\\
\noindent 5. \emph{Separability}: Assume that $\pi(\bx) = \prod_{i=1}^n f_i(x_i)$ and let ${f_i}_\lambda$ be the $\lambda$-Moreau approximation of the marginal density $f_i$. Then $\pi_\lambda(\bx) = \prod_{i=1}^n {f_i}_\lambda(x_i)$.\\
\noindent 6. \emph{Exponential tails}: Assume that $\pi \in \mathcal{E}(\beta,\gamma)$ with $\beta \geq 1$. Then $\pi_\lambda \in \mathcal{E}(\beta^\prime,\gamma^\prime)$ with $\beta^\prime = \min(\beta,2)$.

Properties 1--5 are extensions of well known results for Moreau--Yoshida envelope functions first established in \cite{Moreau1962}. Property 1 results from the fact that in the limit $\lambda \rightarrow 0$ the term $\exp{\left(-\|\bu - \bx\|^2/2\lambda\right)}$ tends to a Dirac delta function at $\bx$. Property 2 can be easily established by using the results of Section 2.3 of \citet{Combettes2005}. Property 3 follows from the fact that $\prox^\lambda_{g}(\bx)$ is the maximiser of $h(\bu) = \log\pi(\bu) - \|\bu - \bx\|^2/2\lambda$ and therefore $0 \in \partial h\{\prox^\lambda_{g}(\bx)\}$ \citep[Lemma 2.5]{Combettes2005}. Property 4 follows from Properties 2 and 3: if $\bx^*$ is a maximiser of $\pi_\lambda$ then from Property 2, $\prox^\lambda_{\pi}(\bx^*) = \bx^*$, and from Property 3, $0 \in \partial\log\pi(\bx^*)$. Then, Fermat's rule, generalised to subdifferentials, together with the fact that $\pi$ is log-concave implies that $\bx^*$ is a maximiser of $\pi$. Property 5 results from the fact that the proximity mapping of the separable sum $g(\bx) = \sum_{i=1}^n \log f_i(x_i)$ is given by $\{\prox^\lambda_{\log f_1}(x_1),\ldots,\prox^\lambda_{\log f_n}(x_n)\}$ \citep[Ch. 2]{BoydBook}. Finally, to establish Property 6 we use \eqref{MoreauApprox} and \eqref{classE} and note that for $\bx$ sufficiently large, $\pi_\lambda$ has exponentially decreasing tails with exponent $\beta^\prime = \beta$ if $\beta \in [1,2]$ and $\beta^\prime = 2$ if $\beta > 2$ (distributions with $\beta < 1$ are not log-concave).

To illustrate these definitions, Fig. \ref{fig:Moreau} depicts the Moreau approximations of four distributions that are log-concave: the Laplace distribution $\pi(x) \propto \exp{\left(-|x|\right)}$, the Gaussian distribution $\pi(x) \propto \exp{\left(-x^2\right)}$, the quartic or fourth-order polynomial distribution $\pi(x) \propto \exp{\left(-x^4\right)}$, and the uniform distribution $\pi(x) \propto \boldsymbol{1}(\bx)_{[-1,1]}$. We observe that the approximations are smooth, converge to $\pi$ as $\lambda$ decreases, and have the same maximisers as the true densities, as described by Properties 1, 2 and 4. We also observe that for densities with lighter-than-Gaussian tails the Moreau approximation mimics the true density around the mode but has Gaussian tails, as described by Property 6. %Section \ref{PPMCMC} shows that this property of Moreau approximations leads to MCMC algorithms with robust stability properties, as opposed to many other high-dimensional algorithms that perform badly for densities with light tails.

\begin{figure}[h!]
\begin{minipage}[l1]{.5\linewidth}
  \centering
  %\centerline{\includegraphics[width=8.0cm]{LaplaceExample.png}}
     \centerline{\includegraphics[width=8.0cm]{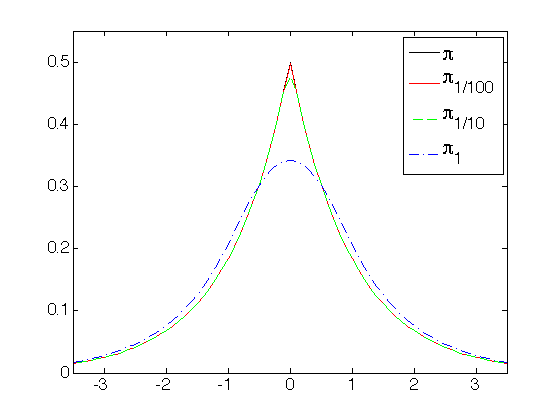}}%F1_Laplace
  \small{(a) $\pi(x) \propto \exp{\left(-|x|\right)}$}
\end{minipage}
\begin{minipage}[l2]{.5\linewidth}
  \centering
  %\centerline{\includegraphics[width=8.0cm]{UniformExample.png}}
   \centerline{\includegraphics[width=8.0cm]{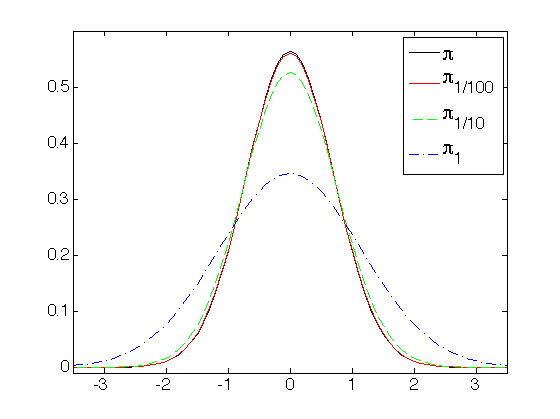}}
    \small{(b) $\pi(x) \propto \exp{\left(-x^2\right)}$}
\end{minipage}
\begin{minipage}[l1]{.5\linewidth}
  \centering
  %\centerline{\includegraphics[width=8.0cm]{LaplaceExample.png}}
     \centerline{\includegraphics[width=8.0cm]{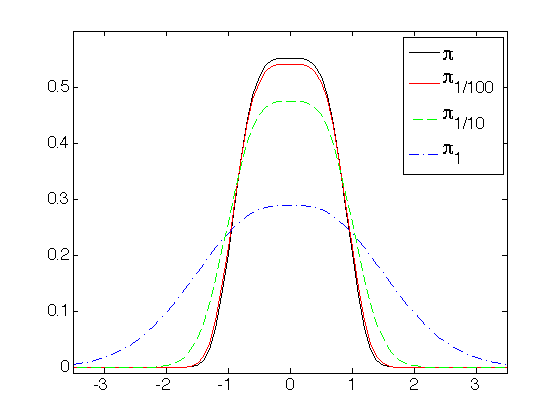}}
  \small{(c) $\pi(x) \propto \exp{\left(-x^4\right)}$}
\end{minipage}
\begin{minipage}[l2]{.5\linewidth}
  \centering
  %\centerline{\includegraphics[width=8.0cm]{UniformExample.png}}
   \centerline{\includegraphics[width=8.0cm]{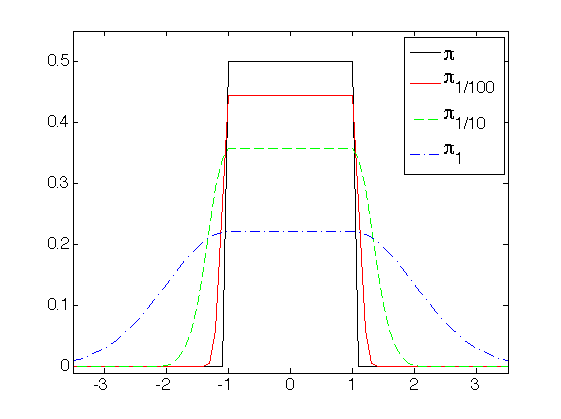}}
    \small{(d) $\pi(x) \propto \boldsymbol{1}(\bx)_{[-1,1]}$}
\end{minipage}
\caption{\small{Density plots for the Laplace (a), Gaussian (b), quartic (c) and uniform (d) distributions (solid black), and their Moreau approximations \eqref{MoreauApprox} for $\lambda= 1, 0.1, 0.01$ (dashed blue and green, and solid red).}} \label{fig:Moreau}
\end{figure}

As mentioned previously, the methods proposed in this paper are useful for models that have proximity mappings which are easy to evaluate or to approximate numerically (see Section \ref{sss:ie3} for more details). This is the case for many statistical models used in high-dimensional data analysis, where statistical inference is often conducted using convex optimisation algorithms that also require computing proximity mappings (see \citet{Figueiredo2011, Candes2009, Recht, Chandrasekaran2012} for examples in image restoration, compressive sensing, low-rank matrix recovery and graphical model selection). For more details about the evaluation of these mappings, their properties, and lists of functions with known mappings please see \citet{CombettesBook}, \citet{Combettes2011} and \citet[Ch. 6]{BoydBook}. A library with MATLAB implementations of frequently used proximity mappings is available on \url{https://github.com/cvxgrp/proximal}. 

\subsection{Langevin Markov chain Monte Carlo}\label{ss:ULAMALA}
The sampling method presented in this paper is derived from the Langevin diffusion process and is related to other Langevin MCMC algorithms that we briefly recall below. 

Suppose that $\pi$ is everywhere non-zero and differentiable so that $\nabla\log\pi$ is well defined. Then let $W$ be the $n$-dimensional Brownian motion and consider a Langevin diffusion process $\{Y(t): 0 \leq t \leq T\}$ on $\mathbb{R}^n$ that has $\pi$ as stationary distribution. Such process is defined as the solution to the stochastic differential equation
\begin{equation}\label{diffusion}
dY(t) = \frac{1}{2}\nabla\log\pi\{Y(t)\}dt + dW(t), \quad Y(0) = y_0.
\end{equation}
Under appropriate stability conditions, $Y(t)$ converges in distribution to $\pi$ and is therefore potentially interesting for simulating from $\pi$. Unfortunately, direct simulation from $Y(t)$ is only possible in very specific cases. A more general solution is to consider a discrete-time approximation of the Langevin diffusion process with step-size $\delta$. For computational reasons a forward Euler approximation is typically used, resulting in the so-called ULA
\begin{equation}\label{ULA}
\textrm{ULA}:\,\quad L^{(m+1)} = L^{(m)} + \frac{\delta}{2}\nabla\log\pi\{L^{(m)}\} + \sqrt{\delta} Z^{(m)}, \quad Z^{(m)} \sim \mathcal{N}(0,\mathbb{I}_n)
\end{equation}
where the parameter $\delta$ controls the discrete-time increment as well as the variance of the Gaussian perturbation $Z^{(m)}$. Under certain conditions on $\pi$ and $\delta$, ULA produces a good approximation of $Y(t)$ and converges to an ergodic measure which is close to $\pi$. In MALA this approximation error is corrected by introducing a  rejection step that guarantees convergence to the correct target density $\pi$ \citep{Roberts1996b}.
 
It is well known that MALA can be a very efficient sampling method, particularly in high-dimensional problems. However, it is also known that for certain classes of target densities ULA is transient and as a result MALA is not geometrically ergodic \citep{Roberts1996b}. Geometric ergodicity is important theoretically to guarantee the existence of a central limit theorem for the chains and practically because sub-geometric algorithms often fail to explore the parameter space properly. Another limitation of MALA and HMC methods is that they require $\pi \in \mathcal{C}^1$. This limits their applicability in many popular image processing and machine models that are not smooth.
 
In the following section we present a new MALA method that use proximity mappings and Moreau approximations to capture the log-concavity of the target density and construct chains with significantly better geometric convergence properties. We emphasise at this point that this is not the first work to consider modifications of MALA with better geometric convergence properties. For example, \citet{Roberts1996b} suggested using MALA with a truncated gradient to retain the efficiency of the Langevin proposal near the density's mode and add robustness in the tails, though we have found this approach to be difficult to implement practically (this is illustrated in Section \ref{sss:ie2}). Also, \citet{Casella2010} recently proposed three variations of MALA based on implicit discretisation schemes that are geometrically ergodic for one-dimensional distributions with super-exponential tails. For certain one-dimensional densities the methods presented in this paper are closely related to the partially implicit schemes of \citet{Casella2010}. Manifold MALA \citep{Girolami2011} is also geometrically ergodic for a wide range of tail behaviours if $\delta$ is sufficiently small \citep{Krys2011}. \black

\section{Proximal MCMC}\label{PPMCMC}
\subsection{Proximal unadjusted Langevin algorithm}\label{ss:pula}
This section presents a proximal Metropolis-adjusted Langevin algorithm (P-MALA) that exploits convex analysis to sample efficiently from log-concave densities $\pi$ of the form \eqref{pi}. In order to define this algorithm we first introduce the \emph{proximal unadjusted Langevin algorithm} (P-ULA) that generates samples approximately distributed according to $\pi$, and that will be used as proposal mechanism in P-MALA. We establish that P-ULA is geometrically ergodic in many cases for which ULA is transient or explosive and that P-MALA inherits these favourable properties, converging geometrically fast in many cases in which MALA does not.%is not geometric.

A key element of this paper is to first approximate the Langevin diffusion $Y(t)$ with an auxiliary diffusion $Y_\lambda(t)$ that has invariant measure $\pi_\lambda$, defined by the stochastic differential equation \eqref{diffusion} with $\pi$ replaced by its $\lambda$-Moreau approximation \eqref{MoreauApprox}. The regularity properties of $\pi_\lambda$ will lead to discrete approximations with favourable stability and convergence qualities. We wish to use $Y_\lambda(t)$ to simulate from $\pi_\lambda$, which we can make arbitrarily close to $\pi$ by selecting a small value of $\lambda$. Direct simulation from $Y_\lambda(t)$ is typically infeasible and we thus consider the forward Euler approximation \eqref{ULA} for $Y_\lambda(t)$,
\begin{equation}\label{langevinDiscrete}
Y^{(m+1)} = Y^{(m)} + \frac{\delta}{2}\nabla\log\pi_\lambda\{Y^{(m)}\} + \sqrt{\delta} Z^{(m)}, \quad Z^{(m)} \sim \mathcal{N}(0,\mathbb{I}_n).
\end{equation}
From Property 2 we obtain that \eqref{langevinDiscrete} is equal to
\begin{equation}\label{langevinDiscrete2}
Y^{(m+1)} = \left(1-\frac{\delta}{2\lambda}\right)Y^{(m)} + \frac{\delta}{2\lambda}\prox^\lambda_g\{Y^{(m)}\} + \sqrt{\delta} Z^{(m)}, \quad Z^{(m)} \sim \mathcal{N}(0,\mathbb{I}_n).
\end{equation}
This Markov chain has two interpretations that provide insight on how to select an optimal value for $\lambda$. First, \eqref{langevinDiscrete2} is a discrete approximation of a Langevin diffusion with invariant measure $\pi_\lambda$, and since we are interested is simulating from $\pi$, we should set $\lambda$ to as small a value as possible to bring $\pi_\lambda$ close to $\pi$. Second, from a convex optimisation viewpoint, \eqref{langevinDiscrete2} coincides with a \emph{relaxed proximal point} iteration to maximise $\log\pi$ with relaxation parameter $\delta/2\lambda$, plus a stochastic perturbation given by $\sqrt{\delta}Z$ \citep{Rockafellar}. According to this second interpretation $\lambda$ should not be smaller than $\delta/2$, as this could lead to an unstable proximal point update that is expansive and therefore to an explosive Markov chain. We therefore define the optimal $\lambda$ as the smallest value within the range of stable values $[\delta/2,\infty)$. Setting $\lambda = \delta/2$ we obtain the P-ULA Markov chain
\begin{equation}\label{langevinDiscrete3}
\textrm{P-ULA}:\,\quad Y^{(m+1)} = \prox^{\delta/2}_{g}\{Y^{(m)}\} + \sqrt{\delta} Z^{(m)}, \quad Z^{(m)} \sim \mathcal{N}(0,\mathbb{I}_n).
\end{equation}

We now study the convergence properties of P-ULA. In a manner akin to \cite{Roberts1996b}, we study geometric convergence for the case where $\pi$ is one-dimensional and we illustrate our results on the class $\mathcal{E}(\beta,\gamma)$. Extensions to high-dimensional models of the form $\pi(\bx) = \prod_{i=1}^n f_i(x_i)$ are possible by using Property 5, and to high-dimensional densities $\pi \in \mathcal{C}^\infty$ with Lipschitz gradients by using Theorem 7.1 of \citet{Mattingly2002}.

\begin{theo}{\label{Theo1}
\textit{Suppose that $\pi$ is one-dimensional and that \eqref{pi} holds. For some fixed $d >0$, let
$$
S_d^+ = \lim_{x\rightarrow\infty} \{\prox^{\delta/2}_{g}(x) - x\}x^{-d}, \quad S_d^- = \lim_{x\rightarrow-\infty} \{\prox^{\delta/2}_{g}(x) - x\}|x|^{-d}.
$$
Then $\textrm{P-ULA}$ is geometrically ergodic if for some $d \in [0,1]$ both $S_d^+ < 0$ and $S_d^- > 0$ exist.
%one of the following holds:\\
%i) for some $d \in [0,1)$ both $S_d^+ < 0$ and $S_d^- > 0$ exist\\
%ii) for $d = 1$ both $S_d^+ < 0$ and $S_d^- > 0$ exist and $(S_d^+ - 1)(1-S_d^-) < 1$.
}}
\end{theo}
\noindent 
\begin{proof}
The proof follows from the fact that $\nabla\log\pi_{\delta/2}$ is continuous and $\textrm{P-ULA}$ is $\mu^{Leb}$-irreducible and weak Feller, and hence all compact sets are small \citep[Ch. 6]{MeynTweedieBook}. Then, using Property 2, the conditions on $S_d^+$ and $S_d^-$ are equivalent to the conditions of part (a) of Theorem 3.1 of \citet{Roberts1996b} establishing that $\textrm{P-ULA}$ is geometrically ergodic for $d \in [0,1)$. For $d = 1$ we proceed similarly to Property 6 and note that for approximations $\pi_{\delta/2}$ with Gaussian tails we have that $S_1^+ \in (-1,0)$ and $S_1^- \in (0,1)$, thus part (b) of Theorem 3.1 of \citet{Roberts1996b} applies. Finally, notice from Property 2 that the values of $d$, $S_d^+$ and $S_d^-$ are closely related to the tails of the approximation $\pi_{\delta/2}$, i.e., $\lim_{x\rightarrow\infty} \frac{\textrm{d}}{\textrm{d}x} \log\pi_{\delta/2}(x) =  S_d^+ x^d + o(|x|^d)$ and $\lim_{x\rightarrow-\infty} \frac{\textrm{d}}{\textrm{d}x}\log\pi_{\delta/2}(x) =  S_d^- x^d + o(|x|^d)$.
\end{proof}

Theorem \ref{Theo1} is most clearly illustrated when $\pi$ belongs to the class $\mathcal{E}(\beta,\gamma)$. Recall that ULA is not ergodic for if $\beta >2$ and only for $\delta$ sufficiently small if $\beta = 2$ \citep{Roberts1996b}.%, which is used in \cite{Roberts1996b} to study the convergence properties of ULA and in \cite{Casella2010} for variations of ULA based on implicit discretisation schemes. We will show that P-ULA has remarkable convergence properties, it is geometrically ergodic for all the distributions in $\mathcal{E}(\beta,\gamma)$ for which ULA is ergodic and for many more for which ULA is transient, regardless of the value of the discretisation step $\delta$.
\begin{coro}{\label{Theo2}
\textit{Assume that $\pi \in \mathcal{E}(\beta,\gamma)$ and that \eqref{pi} holds. Then $\textrm{P-ULA}$ is geometrically ergodic for all $\delta > 0$.}}
\end{coro} 
\noindent This result follows from the fact that \eqref{pi} implies $\beta \geq 1$ (distributions belonging to $\mathcal{E}(\beta,\gamma)$ with $\beta < 1$ are not log-concave), which in turn implies that $\pi_{\delta/2} \in \mathcal{E}(\beta^\prime,\gamma^\prime)$ with $\beta^\prime = \min(\beta,2)$ and some $\gamma^\prime > 0$. The geometric convergence of $\textrm{P-ULA}$ is then established by checking that for $d = \beta^\prime - 1$ the limits $S_d^+$ and $S_d$ exist and verify the conditions of Theorem \ref{Theo1} for all $\delta >0$. %Notice that Corollary \ref{Theo2} indicates very robust theoretical properties for P-ULA; ULA is not ergodic for any $\delta$ if $\beta >2$ and only for $\delta$ sufficiently small if $\beta = 2$ \citep{Roberts1996b}.% and manifold ULA that is only geometrically ergodic if $\delta$ is sufficiently small. %This superior performance results from the regularity properties of $\pi_\lambda$ that ensure that the first order Euler approximation of $y_{\delta/2}$ is always stable and geometrically ergodic. %Precisely, that $\pi_\lambda$ is always continuously differentiable and has at most quadratic tails with scale depending on $\delta$.

The results presented above establish that under certain conditions on $\pi$ $\textrm{P-ULA}$ converges geometrically to some unknown ergodic measure. To determine if this stationary measure is a good approximation of $\pi$, and thus if P-ULA is a good proposal for a  algorithm, we consider the more general question of how well $\textrm{P-ULA}$ approximates the time-continuous diffusion $Y(t)$ as a function of $\delta$ [we consider strong mean-square convergence to $Y(t)$ in the sense of \cite{Higham2003}, which also implies the convergence of P-ULA's ergodic measure to $\pi$]. %This will be relevant for using P-ULA to drive a  algorithm. Note that the following results are valid for general multidimensional models.

\begin{theo}{\label{Theo4}
\textit{Suppose that $\pi \in \mathcal{C}^2$ and that \eqref{pi} holds. Then there exists a continuous-extension $\bar{Y}(t)$ of the $\textrm{P-ULA}$ chain for which
$$
\lim_{\delta \rightarrow 0} \mathbb{E}\left( \sup_{0 \leq t \leq T}  \left|\bar{Y}(t)-Y(t)\right|^2\right) = 0
$$
where $Y(t)$ is the Langevin diffusion \eqref{diffusion} with ergodic measure $\pi$. Moreover, if $\nabla\log\pi$ is polynomial in $\bx$, then $\textrm{P-ULA}$ converges strongly to $Y(t)$ at optimal rate; that is,
$$
\mathbb{E}\left( \sup_{0 \leq t \leq T}  \left|\bar{Y}(t)-Y(t)\right|^2\right) = \textrm{O}(\delta).
$$
}}
\end{theo}
\begin{proof} To prove the first result we use Property 3 to express P-ULA as a \textit{split-step backward Euler} approximation of $Y(t)$ (i.e., $Y^{(m+1)} = Y^{+} + \sqrt{\delta} W^{(m)}$ with $Y^{+} = \frac{\delta}{2}\nabla\log\pi\left(Y^{+}\right) + Y^{(m)}$), and apply Theorem 3.3 of \cite{Higham2003}, where we note that assumption \eqref{pi} implies condition 3.1 of \cite{Higham2003}. The second result follows from  Theorem 4.7 of \cite{Higham2003}.\end{proof}

\subsection{Proximal Metropolis-adjusted Langevin algorithm}\label{ss:pmala}
\subsubsection{ correction}
As explained previously, P-ULA simulates samples from an approximation of $\pi$. A natural strategy to correct this approximation error is to supplement P-ULA with a Metropolis--Hasting accept--reject step guaranteeing convergence to $\pi$, leading to a \emph{proximal Metropolis-adjusted Langevin algorithm} (P-MALA). This is a  chain $X^{(m)}$ that uses $\textrm{P-ULA}$ as proposal. Precisely, given $X^{(m)}$, a candidate $Y^{*}$ is generated by using one $\textrm{P-ULA}$ transition
\begin{equation}\label{PULA_Proposal}
Y^{*} | X^{(m)} \sim \mathcal{N}\left[\prox^{\delta/2}_{g}\{X^{(m)}\}, \delta\mathbb{I}_n\right].
\end{equation}
We accept this candidate and set $X^{(m)} = Y^{*}$ with probability
\begin{equation}\label{ratio2}
\textrm{r}\{X^{(m)},Y^*\} = \min\left[1,\frac{\pi(Y^*)}{\pi\{X^{(m)}\}}\frac{q\{X^{(m)}|Y^*\}}{q\{Y^*|X^{(m)}\}}\right]
\end{equation}
where $q\{Y^*|X^{(m)}\} = p_{\mathcal{N}}\left[Y^*|\prox^\lambda_g\{X^{(m)}\},\delta\mathbb{I}_n\right]$ is the $\textrm{P-ULA}$ transition kernel given by \eqref{langevinDiscrete3}. Otherwise, with probability $1-\textrm{r}\{X^{(m)},Y^*\}$, we reject the proposition and set $X^{(m+1)} = X^{(m)}$. By the Hastings construction, the P-MALA chain converges to $\pi$ in the total-variation norm [this follows from the facts that the chain is irreducible, aperiodic and $\pi$-invariant \citep[ch. 7]{Robert}]. Note that though \eqref{ratio2} involves two proximity mappings, we only need to evaluate $\prox^{\delta/2}_{g}(X^{*})$ at each iteration since $\prox^{\delta/2}_{g}\{X^{(m)}\}$ is known from the algorithm's previous iteration.

%In the remainder of this section we study the geometric ergodicity properties of $\textrm{P-MALA}$.
\subsubsection{Convergence properties}
We provide two alternative sets of conditions for the geometric ergodicity of $\textrm{P-MALA}$ and illustrate our results on the case where $\pi$ belongs to the class $\mathcal{E}(\beta,\gamma)$, which we use as benchmark for comparison with other MALAs.
\begin{theo}{\label{Theo5}
\textit{Suppose that \eqref{pi} holds. Let $A(\bx) = \{\bu: \textrm{r}(\bx,\bu) = 1 \}$ be the acceptance region of $\textrm{P-MALA}$ from point $\bx$, and $I(\bx) = \{\bu: \|\bx\| \geq \|\bu\| \}$ the region of points interior to $\bx$. Suppose that $A$ converges inwards in $q$, i.e., 
$$
\lim_{\|\bx\|\rightarrow\infty} \int_{A(\bx)\Delta I(\bx)} q(\bu|\bx)\mathrm{d}\bu = 0
$$
where $A(\bx)\Delta I(\bx)$ denotes the symmetric difference  $A(\bx) \cup I(\bx) \setminus A(\bx) \cap I(\bx)$. Then $\textrm{P-MALA}$ is geometrically ergodic.
}}
\end{theo} 
\begin{proof} To prove this result we use Theorem 5.14 of \citet{CombettesBook} to show that if \eqref{pi} holds then, for any $\bx$, the mean candidate position $\prox^\lambda_g(\bx)$ verifies the inequality $\|\prox^\lambda_g(\bx)\| < \|\bx\|$. This result, together with the condition that $A$ converges inwards in $q$, implies that $\textrm{P-MALA}$ is geometrically ergodic \citep[Theorem 4.1]{Roberts1996b}.
\end{proof}
\begin{coro}{\label{Theo7}
\textit{Suppose that $\pi \in \mathcal{E}(\beta,\gamma)$ and that \eqref{pi} holds. Then P-MALA is geometrically ergodic for all $\delta > 0$.}}
\end{coro} 
\noindent Proving this result simply consists of checking that if $\pi \in \mathcal{E}(\beta,\gamma)$ and \eqref{pi} holds then $A$ converges inwards in $q$ and therefore Theorem \ref{Theo5} applies, where we note that \eqref{pi} implies that $\beta \geq 1$. 

Notice from Corollary \ref{Theo7} that P-MALA has very robust stability and converge properties. For comparison, MALA is not geometrically ergodic for any $\pi \in \mathcal{E}(\beta,\gamma)$ with $\beta > 2$ \citep{Roberts1996b} and manifold MALA is geometrically ergodic for $\pi \in \mathcal{E}(\beta,\gamma)$ with $\beta \neq 1$ only if $\delta$ is sufficiently small \citep{Krys2011}. P-MALA inherits these robust convergence properties from P-ULA, or more precisely from the regularity properties of $\pi_{\delta/2}$ that guarantee that P-ULA is always stable and geometrically ergodic. In particular, that $\log \pi_{\delta/2}$ decays at mostly quadratically, that $\nabla \log \pi_{\delta/2}$ always exists and is Lipchitz continuous, and that the tails of $\pi_{\delta/2}$ broaden with $\delta$ such that $Y_{\delta/2}(t)$ is always within the stability range of a forward Euler approximation with time step $\delta$. 

Moreover, the convergence properties of $\textrm{P-MALA}$ can also be studied in the framework of Random-walk  algorithms with bounded drift \citep{Atachade2006}. 
\begin{theo}{\label{Theo6}
\textit{Suppose that $\pi \in \mathcal{C}^1$ and that \eqref{pi} holds. Assume that there exists $R > 0$ such that $\forall \bx \in \mathbb{R}^n, \|\bx - \prox^{\delta/2}{g}(\bx)\| < R $, and that $\pi$ verifies the conditions
$$
\lim_{\|\bx\| \rightarrow \infty} \frac{\bx}{\|\bx\|} \cdot \nabla\log\pi(\bx) = - \infty \quad\mbox{and}\quad \lim_{\|\bx\| \rightarrow \infty} \frac{\bx}{\|\bx\|} \cdot \frac{\nabla\pi(\bx)}{\|\nabla\pi(\bx)\|}<0.
$$
Then $\textrm{P-MALA}$ is geometrically ergodic.
}}
\end{theo}
\begin{proof}The proof of this result follows from the proof of geometric ergodicity for the Shrinkage-thresholding MALA \citep{Schreck2013}, which is general to all  algorithms with bounded drift, and where we note that the conditions on $\pi$, together with the bounded drift condition $\|\bx - \prox^\lambda_{g}(\bx)\| < R$, satisfy the assumptions of Theorem 4.1 of \citet{Schreck2013}. 
\end{proof}
Notice that it is always possible to enforce the bounded drift condition by composing $\prox^\lambda_{g}(\bx)$ with a projection onto an $\ell_2$-ball centred at $\bx$ (this is equivalent to using a truncated gradient as proposed in \citep{Roberts1996b}). Also, it is possible to relax the smoothness assumption to $\pi \in \mathcal{C}^0$ by adding assumptions A3 and A4 from \citet{Schreck2013}.

Finally, similarly to other MH algorithms based on local proposals, P-MALA may be geometrically ergodic yet perform poorly if the proposal variance $\delta$ is either too small or very large. Theoretical and experimental studies of MALA show that for many high-dimensional target densities the value of $\delta$ should be set to achieve an acceptance rate of approximately $40\% - 70\%$ \citep{pillai2012}. These results do not apply directly to P-MALA. However, given the similarities between MALA and P-MALA, it is reasonable to assume that the values of $\delta$ that are appropriate for MALA will generally also produce good results for P-MALA. In our experiments we have found that P-MALA performs well when $\delta$ is set to achieve an acceptance rate of $40\% - 60\%$. %An optimal scaling analysis of P-MALA is currently under investigation.

%The results in that work were originally established under the assumption of $\pi \in \mathcal{C}^1$ and have has been recently relaxed to $\pi \in \mathbb{C}^0$ by \citet{Schreck2013}. The results of \citet{Schreck2013} make two key assumptions on $\pi$, which we recall below:
%\noindent\textbf{A2} \citep{Schreck2013}: For any $s>0$
%$$
%\lim_{\rho \rightarrow \infty} \,\sup_{\bx, \|\bx\|>\rho} \frac{\pi(\bx + s n(\bx))}{\pi(\bx)}=0, \quad \textrm{where} \quad n(\bx)\triangleq \frac{\bx}{\|\bx\|}.
%$$
%When $\pi$ is differentiable A2 is satisfied if \citep{Jarner2000341}
%$$
%\lim_{\|\bx\| \rightarrow \infty} n(\bx)\nabla\log\pi(\bx) = - \infty
%$$
%
%\noindent\textbf{A3} \citep{Schreck2013}: For any $s>0$
%There exist $u,b,\epsilon>0$  and $u \in (0,b)$ such that for any $\bx :  \|\bx\| > R$
%$$
%\forall \by \in W(\bx),\quad \pi(\bx + u n(\bx)) < \pi(\by)
%$$
%where $W(\bx)$ is the cone of $\mathbb{R}^n$ with apex $\bx + u n(\bx)$ and aperture $2\epsilon$. This condition guarantees that, for $\|\bx\|$ large enough, the probability to accept a move from $\bx$ to any point in $W(\bx)$ equals one \citep{Schreck2013}. When $\pi$ is differentiable A3 is implied by \citep{Jarner2000341}
%$$
%\lim_{\|\bx\| \rightarrow \infty} \frac{\bx}{\|\bx\|}\frac{\nabla\pi(\bx)}{\|\nabla\pi(\bx)\|}<0.
%$$

\subsubsection{Computation of the proximity mapping $\prox^{\delta/2}_{g}(\bx)$}\label{sss:ie3}
The computational performance of P-MALA depends strongly on the capacity to evaluate efficiently $\prox^{\delta/2}_{g}(\bx) = \argmax_{\bu \in \mathbb{R}^n}\, g(\bu) -\|\bu - \bx\|^2/\delta$. As mentioned previously, the computation of proximity mappings is the focus of significant research efforts because these operators are key to modern convex and non-convex optimisation. As a result, for many important models used in high-dimensional data analysis, signal and image processing, and statistical machine learning, there are now clever analytical or numerical techniques to evaluate these mappings efficiently (two examples of this are the total-variation and the nuclear-norm priors used in the experiments of Section 4). For a survey on the evaluation of proximity mappings and lists of some functions with known mappings please see \citet[Ch.6]{BoydBook} and \citet{Combettes2011}. 

The most general strategy for computing $\prox^{\delta/2}_{g}(\bx)$ is to note that \eqref{proxMap} is a convex optimisation problem that can frequently be solved or approximated quickly with state-of-the-art convex optimisation algorithms. \citet{Pesquet2014} presents these algorithms in the primal-dual framework and provides clear guidelines for parallel and distributed implementations. When applying these techniques within P-MALA it is important to use $\bx$ to \emph{hot-start} the optimisation, particularly in high-dimensional models where $\prox^{\delta/2}_{g}(\bx)$ is close to $\bx$ because $\delta$ has been set to a small value to achieve a good acceptance probability (recall that $\prox^{\delta/2}_{g}(\bx) \rightarrow \bx$ when $\delta \rightarrow 0$).

Alternatively, for many popular models it possible to approximate $\prox^{\delta/2}_{g}(\bx)$ very efficiently by using a decomposition $g(\bx) = g_1(\bx) + g_2(\bx)$ where $g_1 \in \mathcal{C}^1$ is concave with $\nabla g_1$ Lipschitz continuous and where $\prox^{\delta/2}_{g_2}$ can be evaluated efficiently. This enables the approximation
\begin{equation}\label{FBapprox}
\begin{split}
\prox^{\delta/2}_{g}(\bx) 	&= \argmax_{\bu \in \mathbb{R}^n}\, g_1(\bu) + g_2(\bu) -\|\bu - \bx\|^2/\delta \\
					&\approx \argmax_{\bu \in \mathbb{R}^n}\, g_1(\bx) + (\bu-\bx)^T\nabla g_1(\bx)+ g_2(\bu) -\|\bu - \bx\|^2/\delta \\
					& \approx \argmax_{\bu \in \mathbb{R}^n}\, g_2(\bu) -\|\bu - \bx - \delta \nabla g_1^T(\bx)\|^2/\delta \\
					& \approx \prox^{\delta/2}_{g_2}(\bx + \delta \nabla g_1(\bx))
\end{split}
\end{equation}
that is used in the \emph{forward-backward} or \emph{proximal gradient} algorithm \citep{Combettes2011}. We found this approximation to be very accurate for high-dimensional models because, again, $\delta$ is set to a small value and $\prox^{\delta/2}_{g}(\bx)$ is close to $\bx$, and as a result the approximation $g_1(\bu) \approx g_1(\bx) + (\bu-\bx)^T\nabla g_1(\bx)$ is generally accurate. Approximation \eqref{FBapprox} is useful for instance in linear inverse problems of the form $g(\bx) = -(\by-H\bx)^T \Sigma^{-1}(\by-H\bx)/2 -\alpha\phi(\bx)$ involving a  Gaussian likelihood and a convex regulariser $\phi(\bx)$ with a tractable proximity mapping [$\phi(\bx)$ is often some norm, which generally have known and fast proximity mappings \citep[Ch. 6.5]{BoydBook}]. Notice that many signal and image processing problems can be formulated in this way \citep{Combettes2011}. Moreover, if $g_1 \in \mathcal{C}^2$ it is also possible to use a second-order approximation
\begin{equation}\label{FBapprox2}
\prox^{\delta/2}_{g}(\bx) \approx \argmax_{\bu \in \mathbb{R}^n}\, (\bu-\bx)^T\nabla g_1(\bx)+ (\bu-\bx)^T \frac{H(\bx)}{2} (\bu-\bx) + g_2(\bu) -\|\bu - \bx\|^2/\delta
\end{equation}
where $H_{i,j}(\bx) = \partial^2 g_1/\partial x_i \partial x_j$ or an approximation that simplifies the computation of \eqref{FBapprox2} (for example, if $\prox^{\delta/2}_{g_2}$ is separable, then using a diagonal approximation of the Hessian matrix of $g_1$ leads to an approximation \eqref{FBapprox2} that can be computed in parallel for each element of $\bx$, and that has the same computational complexity as \eqref{FBapprox}). Again, many signal and image processing models it is possible to solve \eqref{FBapprox2} efficiently with a few iterations of the ADMM algorithm of \citet{Figueiredo2011}, which exploits the second-order information from $H(\bx)$ to improve convergence speed.

Finally, it is worth noting that although using an approximation of $\prox^{\delta/2}_{g}(\bx)$ can potentially reduce P-MALA's mixing speed, if the conditions for geometric ergodicity of Theorem \ref{Theo6} hold when $\prox^{\delta/2}_{g}(\bx)$ is evaluated exactly, then P-MALA implemented with an approximate mapping also converges geometrically to $\pi$ if the approximation error can be bounded by some $R^\prime > 0$ or if $\prox^\lambda_{g}(\bx)$ is followed by a projection to guarantee a bounded drift.

\subsubsection{Illustrative example}\label{sss:ie2}
For illustration we show an application of P-MALA to the density $\pi(x) \propto \exp(-x^4)$ depicted in Figure \ref{fig:Moreau}(c). We compare our results with MALA, with the truncated gradient MALA (MALTA) \citep{Roberts1996b}, and with the simplified manifold MALA (SMMALA) \citep{Girolami2011}. As explained previously, MALA is not geometrically ergodic for this target density due to the lighter-than-Gaussian tails. This can be cured by using MALTA, which is a bounded-drift random-walk  algorithm constructed by replacing $h(x) = \nabla \log \pi(x)$ in the MALA proposal with $h_{\epsilon_1}(x) = \epsilon_1 h(x) / \max(\epsilon_1,\|h(x)\|)$ for some $\epsilon_1 > 0$ \citep{Atachade2006}. Although geometrically ergodic, MALTA can converge very slowly if the truncation threshold $\epsilon_1$ is not set correctly. Setting good values for $\epsilon_1$ can be difficult in practice, particularly because values that appear suitable in certain regions of the state space are unsuitable in others. Alternatively, manifold MALA implemented using the (regularised) inverse Hessian $H^{-1}_{\epsilon_2}(x) = (12x^{2} + \epsilon_2)^{-1}$ is also geometrically ergodic if $\delta$ is sufficiently small (for this example $\delta \leq 6$) \citep{Krys2011}, however this algorithm can also converge slowly if the value of $\epsilon_2$ is not set properly.%set properly.%too small or too large. %we have found that its performance is very sensitive to the value of $\varepsilon$.

Figures \ref{fig:P-MALA}(a)-(d) display the first $250$ samples of the chains generated with P-MALA, MALA, MALTA and SMMALA with initial state $X^{(0)} = 10$ and $\delta = 1$. We implemented MALTA and SMMALA using the values $\epsilon_1 = 20$ and $\epsilon_2 = 0.1$ that we adjusted during a series of pilot runs. We found that MALTA  behaves like a Random-walk  algorithm for smaller values of $\epsilon_1$, and that for larger values it rejects the proposed moves with very high probability and gets ``stuck''. Similarly,  we found that SMMALA is very sensitive to the value of $\epsilon_2$, with too small values leading to poor mixing around the mode and larger values to poor mixing in the tails. 

We observe in Figures \ref{fig:P-MALA}(a)-(d) that the chains generated with P-MALA and MALTA exhibit good mixing, that SMMALA has slower mixing, and that MALA has rejected all the proposed moves and failed to converge. We repeated this experiment using the initial state $X^{(0)} = 5$ and the same values for $\delta$, $\epsilon_1$ and $\epsilon_2$. The first $250$ samples of each chain are displayed in Figures \ref{fig:P-MALA}(e)-(h). Again, we observe the good mixing of P-MALA, the slower mixing of SMMALA, and the lack of ergodicity of MALA. However, we also observe that in this occasion MALTA got ``stuck'' at states where its mixing properties are very poor and failed to converge. We also repeated this experiment with HMC (not shown) and observed that it suffers from the same drawbacks as MALA.

\begin{figure}[h!]
\begin{minipage}[l1]{.245\linewidth}
  \centering
  %\centerline{\includegraphics[width=8.0cm]{LaplaceExample.png}}
  \centerline{\includegraphics[width=4.5cm]{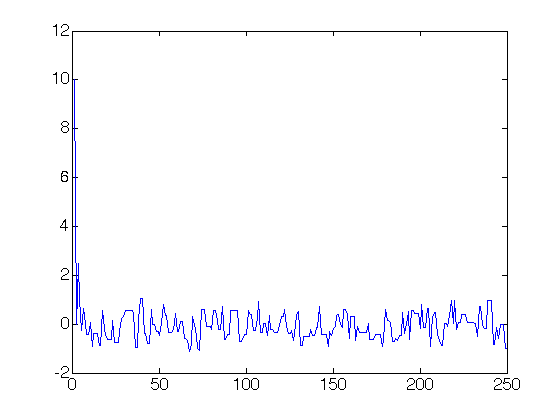}}
  \small{(a) P-MALA}
\end{minipage}
\begin{minipage}[l2]{.245\linewidth}
  \centering
  %\centerline{\includegraphics[width=8.0cm]{UniformExample.png}}
   \centerline{\includegraphics[width=4.5cm]{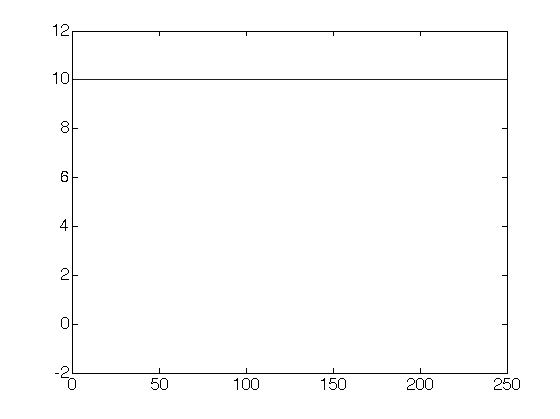}}
    \small{(b) MALA}
\end{minipage}
\begin{minipage}[l3]{.245\linewidth}
  \centering
  %\centerline{\includegraphics[width=8.0cm]{LaplaceExample.png}}
     \centerline{\includegraphics[width=4.5cm]{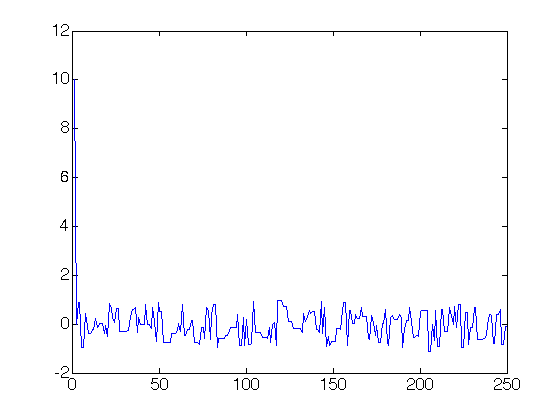}}
     \small{(c) MALTA}
\end{minipage}
\begin{minipage}[l4]{.245\linewidth}
  \centering
  %\centerline{\includegraphics[width=8.0cm]{LaplaceExample.png}}
     \centerline{\includegraphics[width=4.5cm]{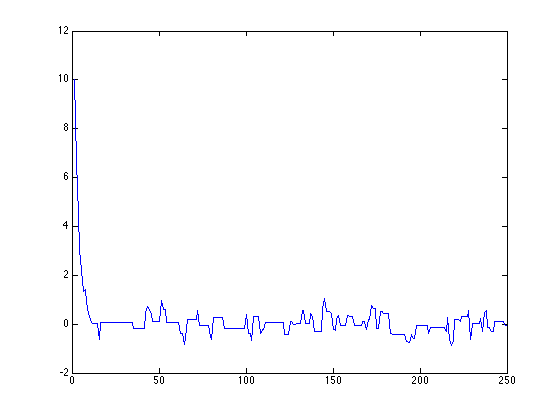}}
     \small{(d) SMMALA}
\end{minipage}

\begin{minipage}[l1]{.245\linewidth}
  \centering
  %\centerline{\includegraphics[width=8.0cm]{LaplaceExample.png}}
     \centerline{\includegraphics[width=4.5cm]{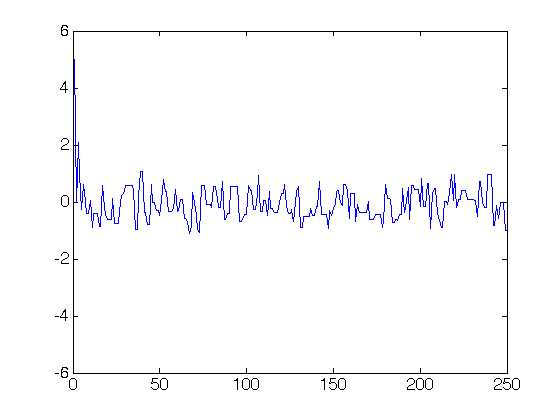}}
  \small{(e) P-MALA}
\end{minipage}
\begin{minipage}[l2]{.245\linewidth}
  \centering
  %\centerline{\includegraphics[width=8.0cm]{UniformExample.png}}
   \centerline{\includegraphics[width=4.5cm]{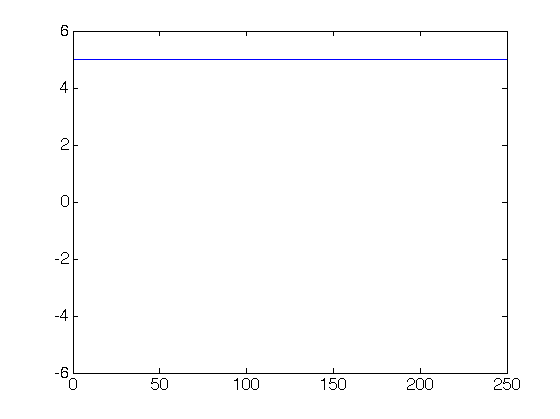}}
    \small{(f) MALA}
\end{minipage}
\begin{minipage}[l1]{.245\linewidth}
  \centering
  %\centerline{\includegraphics[width=8.0cm]{LaplaceExample.png}}
     \centerline{\includegraphics[width=4.5cm]{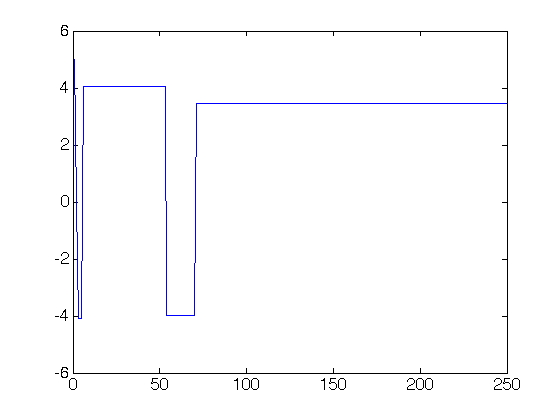}}
  \small{(g) MALTA}
\end{minipage}
\begin{minipage}[l1]{.245\linewidth}
  \centering
  %\centerline{\includegraphics[width=8.0cm]{LaplaceExample.png}}
     \centerline{\includegraphics[width=4.5cm]{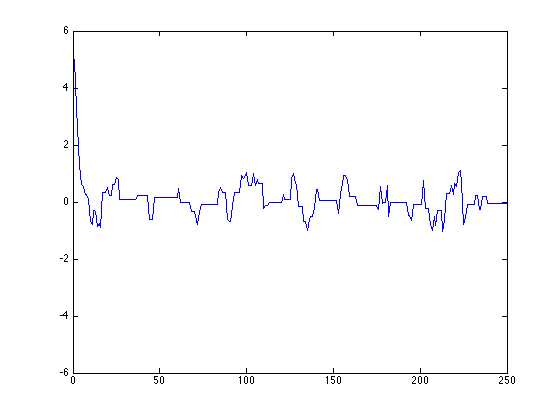}}
  \small{(h) SMMALA}
\end{minipage}
\caption{\small{Comparison between P-MALA, MALA, the truncated gradient MALA (MALTA), and simplified manifold MALA (SMMALA) using the one-dimensional density $\pi(x) \propto \exp\{-x^4\}$ and algorithm parameters $\delta = 1, \epsilon_1 = 20, \epsilon_2 = 0.1$. Initial state $X^{(0)} = 10$ (a)-(d) and $X^{(0)} = 5$ (e)-(h).}} \label{fig:P-MALA}%Trace plots with starting condition $X^{(0)} = 10$ (a)-(c) and $X^{(0)} = 5$ (d)-(f).
\end{figure}

\section{Applications}\label{sec:experiments}
This section demonstrates P-MALA on two challenging high-dimensional and non-smooth models that are widely used in statistical signal and image processing and that are not well addressed by existing MCMC methodology. The first example considers the computation of Bayesian credibility regions for an image resolution enhancement problem. The second example presents a graphical posterior predictive check of the popular \textit{nuclear-norm} model for low-rank matrices. %The second example . %An application of the Bayesian Lasso model\citep{ParkCasella2008} to a variable selection problem and to audio compressive sensing is available in the technical report \red{[X]}.\black.

\subsection{Bayesian image deconvolution with a total-variation prior}
In image deconvolution or deblurring problems, the goal is to recover an original image $\bx \in \mathbb{R}^n$ from a blurred and noisy observed image $\by \in \mathbb{R}^n$ related to $\bx$ by the linear observation model\footnote{note that bidimensional and tridimensional images can be represented as points in $\mathbb{R}^n$ via lexicographic ordering.} $\by = H\bx + \bw$,
where $H$ is a linear operator representing the blur point spread function and $\bw$ is the sample of a zero-mean white Gaussian vector with covariance matrix $\sigma^2\boldsymbol{I}_n$ \citep{hansen:2006}. This inverse problem is usually ill-posed or ill-conditioned, i.e., either $H$ does not admit an inverse or it is nearly singular, thus yielding highly noise-sensitive solutions. Bayesian image deconvolution methods address this difficulty by exploiting prior knowledge about $\bx$ in order to obtain more robust estimates. One of the most widely used image priors for deconvolution problems is the improper total-variation norm prior, $\pi(\bx) \propto \exp{\left(-\alpha\|\nabla_d \bx\|_1\right)}$, where $\nabla_d$ denotes the discrete gradient operator that computes the vertical and horizontal differences between neighbour pixels. This prior encodes the fact that differences between neighbour image pixels are often very small and occasionally take large values (i.e., image gradients are nearly sparse). Based on this prior and on the linear observation model described above, the posterior distribution for $\bx$ is given by
\begin{equation}\label{deconvolution}
\pi(\bx|\by) \propto \exp{\left[-\|\by-H\bx\|^2/2\sigma^2 -\alpha \|\nabla_d\bx\|_1\right]}.
\end{equation}
Image processing methods using \eqref{deconvolution} are almost exclusively based on maximum-a-posteriori (MAP) estimates of $\bx$ that can be efficiency computed using proximal optimisation algorithms \citep{Figueiredo2011}. Here we consider the problem of computing credibility regions for $\bx$, which we use to assess the confidence in the restored image. Precisely, we note that \eqref{deconvolution} is log-concave and use P-MALA to compute marginal $90\%$ credibility regions for each image pixel. There are several computational strategies for evaluating the proximity mapping of $g(\bx) = -\|\by-H\bx\|^2/2\sigma^2 -\alpha \|\nabla\bx\|_1$. Here we take advantage of the fact that in high-dimensional scenarios $\delta$ is typically set to a small value and use the approximation \eqref{FBapprox} $\prox^{\delta/2}_{g}(\bx) \approx \prox^{\delta/2}_{g_2}\{\bx + \delta\nabla g_1(\bx)/2\}$ with $g_1(\bx) = -\|\by-H\bx\|^2/2\sigma^2$ and $g_2(\bx) = -\alpha \|\nabla\bx\|_1$, and where we note that $\nabla g_1$ is Lipschitz continuous and that $\prox^{\delta/2}_{g_2}(\bx)$ can be efficiently computed using a parallel implementation of \citet{Chambolle}. 

Figure \ref{FigCameraman} presents an experiment with the ``cameraman'' image, which is a standard image to assess deconvolution methods \citep{Oliveira2009}. Figures \ref{FigCameraman}(a) and (b) show the original cameraman image ${\bx}_0$ of size $128\times128$ and a blurred and noisy observation $\by$, which we produced by convoluting ${\bx}_0$ with a uniform blur of size $9\times9$ and adding white Gaussian noise to achieve a blurred signal-to-noise ratio (BSNR) of $40$dB ($BRSN = 10\log_{10}\{\textrm{var}(H{\bx}_0)/\sigma^2\}$). The MAP estimate of $\bx$ obtained by maximising \eqref{deconvolution} is depicted in Figure \ref{FigCameraman}(c). This estimate has been computed with the proximal optimisation algorithm of \citet{Figueiredo2011}, and by using the technique of \citet{Oliveira2009} to determine the value of $\alpha$. By comparing Figures \ref{FigCameraman}(a) and \ref{FigCameraman}(c) we observe that the MAP estimate is very accurate and that it restored the sharp edges and fine details in the image. Finally, Figure \ref{FigCameraman}(d) shows the magnitude of the marginal $90\%$ credibility regions for each pixels, as measured by the distance between the $95\%$ and $5\%$ quantile estimates. These estimates were computed from a $20\,000$-sample chain generated with P-MALA using a thinning factor of $1\,000$ to reduce the algorithm's memory foot-print and $1$ million burn-in iterations. These credibility regions show that there is a background level of uncertainty of about $30$ grey-levels, which is approximately $10\%$ of the dynamic range of the image ($256$ grey-levels). More importantly, we observe that there is significantly more uncertainty concentrated at the contours and object boundaries in the image. This reveals that model \eqref{deconvolution} is able to accurately detect the presence of sharp edges in the image but with some uncertainty about their exact location. Therefore computing credibility regions could be particularly relevant in applications that use images to determine the location and the size of objects, or to compare the size of a same object appearing in two different images. For example, in oncological medical imaging, where deconvolution is increasingly used to improve the resolution of images that are subsequently used to assess the evolution of tumour boundaries over time and make treatment decisions. %The $90\%$ credibility regions computed with the variation of MALA described above (results not shown here) coincide with those of Figure \eqref{FigCameraman}(d). 

Moreover, to asses the efficiency of P-MALA we repeated the experiment with a variation of MALA for partially non-differentiable target densities that uses only the gradient of the differentiable term of \eqref{deconvolution}, i.e., $\nabla\log g_1 (\bx)= H^T(\by-H\bx)/\sigma^2$ (this variation of MALA was recently used in \cite{Schreck2013} for a Bayesian variable selection problem with a Bernulli--Laplace prior that is also non-differentiable). Figure \ref{CameramanACF} compares the first $20$ lags of the sample autocorrelation function of the chains generated with P-MALA and MALA, computed using $\log\pi(\bx|\by)$ as scalar summary. We observe that the chain produced with P-MALA has significantly lower autocorrelation and therefore higher effective sample size\footnote{Recall that $ESS = N\{1+2\sum_k \gamma(k)\}^{-1}$, where $N$ is the total  samples and $\sum_k \gamma(k)$ is the sum of the $K$ monotone sample auto-correlations which we estimated with the initial monotone sequence estimator \citep{Geyer1992}} (ESS). P-MALA was almost twice as computationally expensive as MALA due to the overhead associated with evaluating the proximity mapping of $g_2$ (the total computational times were $49$ hours for P-MALA and $28$ hours for MALA). However, because P-MALA is exploring the parameter space significantly faster than MALA, its time-normalised ESS was $4.5$ times better than that of MALA ($50.8$ and $11.04$ samples per hour respectively), confirming the good performance of the proposed methodology. Preconditioning MALA with the (regularised) inverse Fisher information matrix $(H^T H + \epsilon\mathbb{I}_n)^{-1}$ led to poor mixing, possibly because most of the correlation structure in the posterior distributions comes from the non-differentiable prior $\pi(\bx) \propto \exp{[-\alpha \|\nabla_d\bx\|_1]}$ and is not captured by this metric.

\begin{figure}[htbp!]
\begin{minipage}[l2]{.5\linewidth}
  \centering
  \centerline{\includegraphics[width=8.2cm]{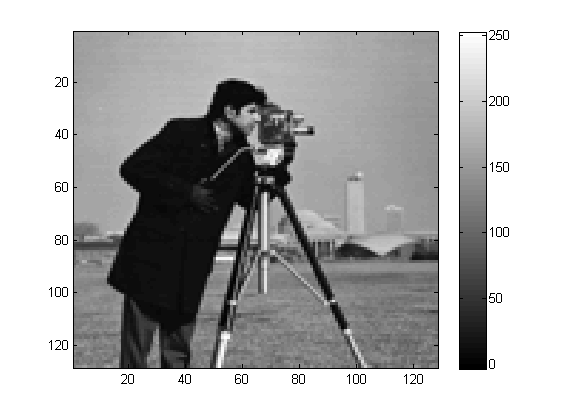}}
  \small{(a)}
\end{minipage}
\begin{minipage}[l2]{.5\linewidth}
  \centering
  \centerline{\includegraphics[width=8.2cm]{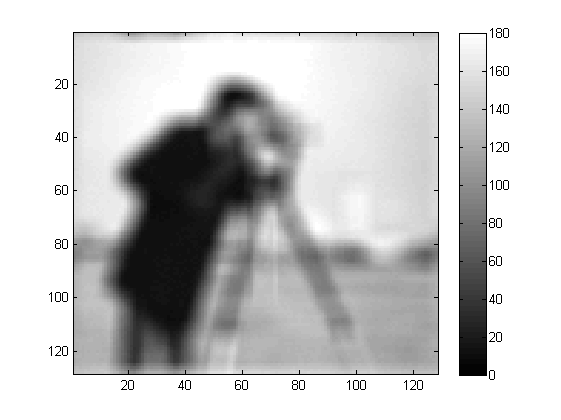}}
  \small{(b)}
\end{minipage}
\begin{minipage}[l2]{.5\linewidth}
  \centering
  \centerline{\includegraphics[width=8.2cm]{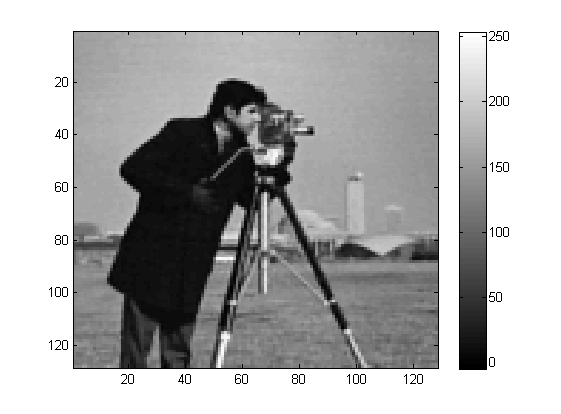}}
  \small{(c)}
\end{minipage}
\begin{minipage}[l2]{.5\linewidth}
  \centering
  \centerline{\includegraphics[width=8.2cm]{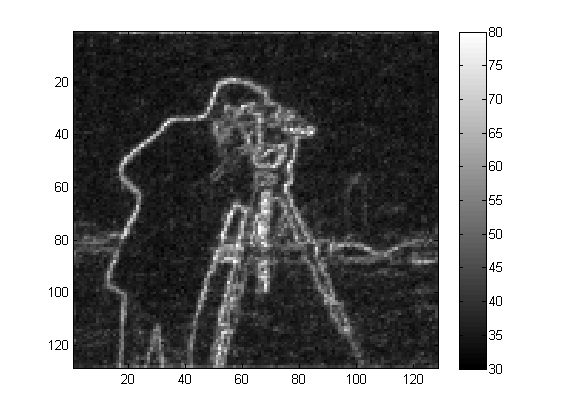}}
  \small{(d)}
\end{minipage}
\caption{\small{(a) Original cameraman image ($128\times128$ pixels), (b) Blurred image, (c) MAP estimate computed with \citep{Figueiredo2011}, (d) Pixel-wise $90\%$ credibility intervals estimated with P-MALA.}} \label{FigCameraman}
\end{figure}

\begin{figure}[htbp!]
\begin{minipage}[l2]{.5\linewidth}
  \centering
  \centerline{\includegraphics[width=6.0cm]{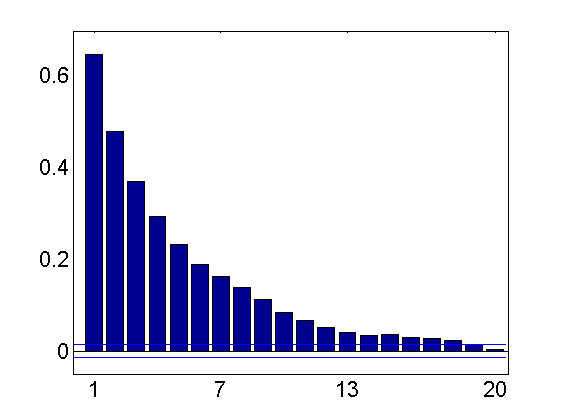}}
  \small{P-MALA}
\end{minipage}
\begin{minipage}[l2]{.5\linewidth}
  \centering
  \centerline{\includegraphics[width=6.0cm]{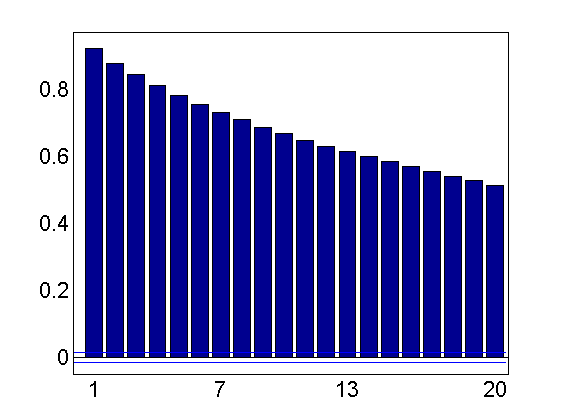}}
  \small{MALA}
\end{minipage}
\caption{\small{Autocorrelation comparison between P-MALA and MALA when simulating from \eqref{deconvolution}.}} \label{CameramanACF}
\end{figure}

\subsection{Nuclear-norm models for low-rank matrix estimation}
In this experiment we use P-MALA to perform a graphical posterior predictive check of the widely used \emph{nuclear norm} model for low-rank matrices \citep{Fazel}. Simulating samples from distributions involving the nuclear norm is challenging because matrices are often high-dimensional and because this norm is not continuously differentiable; thus making it difficult to use gradient-based MCMC methods such as MALA and HMC. For simplicity we present our analysis in the context of matrix denoising, however the approach can be easily applied to other low-rank matrix estimation problems such as matrix completion and decomposition \citep{Chandrasekaran2011,Chandrasekaran2012, Candes2009b, Candes2009c, Candes2011}.

Let $\bx$ be an unknown low-rank matrix of size $n = n_1 \times n_2$ (represented as a point in $\mathbb{R}^n$ by lexicographic ordering), and $\by = \bx + \bw$ a noisy observation contaminated by zero-mean white Gaussian noise with covariance matrix $\sigma^2\boldsymbol{I}_n$. For example, $\bx$ can represent a low-rank covariance matrix in a model selection problem, the background component of a video signal in an object tracking problem, or a rank-limited image in a signal restoration or reconstruction problem \citep{Chandrasekaran2012, Candes2011, Recht}. In the low-rank matrix denoising problem, we seek to recover $\bx$ from $\by$ under the prior knowledge that $\bx$ has low rank; that is, that most of its singular values are zero. A convenient model for this type of problem is the nuclear norm prior $ \pi(\bx) \propto \exp(-\alpha ||\bx||_*)$, where $||\bx||_*$ denotes the nuclear norm of $\bx$ and is defined as the sum of its singular values \citep{Fazel}. The popularity of this prior stems from the fact that the nuclear norm is the best convex approximation of the rank function, and it leads to a posterior distribution that is  log-concave and whose MAP estimate can be efficiently computed using proximal algorithms \citep{Recht}. The posterior distribution of $\bx$ given $\by$ is
\begin{equation}\label{Xposterior}
\pi(\bx|\by) \propto \exp{(-||\by-\bx||^2/2\sigma^2 -\alpha ||\bx||_*)},
\end{equation}
where $\sigma^2$ and $\alpha$ are fixed positive hyper-parameters. It is useful to think of \eqref{Xposterior} as an extension the Bayesian LASSO model \citep{ParkCasella2008} to matrices with sparse singular values, in which the singular values of $\bx$ are assigned exponential priors.% with parameter $\alpha$.

It is well documented that under certain conditions on the true rank and $\sigma^2$, the MAP estimate maximising \eqref{Xposterior} accurately recovers the true null and column spaces of $\bx$ \citep{Candes2009b, Candes2009c, Negahban, Mazumder}. This has led to the general consensus that the nuclear-norm prior is a useful model for low-rank matrix estimation problems and that the errors introduced by using the convex approximation to the rank function do not have a significant effect on the inferences. Here we adopt a Bayesian model checking viewpoint and assess the nuclear-norm model by comparing the observation $\by$ to replicas $\by^{rep}$ generated by drawing samples from the posterior predictive distribution  $f(\by^{rep}|\by) = \int_{\mathbb{R}^{n\times m}} f(\by^{rep}|\bx)\pi(\bx|\bx)\textrm{d}\bx$, as recommended by \citet[Ch. 6]{GelmanBook}. This technique for checking the fit of a model to data is based on the intuition that \emph{``If the model fits, then replicated data generated under the model should look similar to observed data. To put it another way, the observed data should look plausible under the posterior predictive distribution."} \citep[Ch. 6]{GelmanBook}. In this paper we perform a graphical check and compare visually $\by$ and its replicas $\by^{rep}$. In specific applications one could also use $\by^{rep}$ to compute posterior predictive p-values that evaluate specific aspects of the model that are relevant to the application \citep[Ch. 6]{GelmanBook}.

Figure \ref{FigCheckerboard} presents an experiment with MATLAB's ``checkerboard'' image. Figure \ref{FigCheckerboard}(a) shows the original checkerboard image $\bx_0$ of size $n = 64\times 64$ pixels and rank $2$. Figure \ref{FigCheckerboard}(b)  shows a noisy observation $\by$ produced by adding Gaussian noise with variance $\sigma^2 = 0.01$, leading to a signal-to-noise ratio (SNR) of $15$dB which is standard for image denoising problems ($SNR = 10\log_{10}(||\bx_0||^2/n m\sigma^2)$). The MAP estimate obtained by maximising \eqref{Xposterior} is depicted in Figure \ref{FigCheckerboard}(c). This estimate has been computed via singular value soft-thresholding, and by setting $\alpha=1.15/\sigma^2$ to minimise Stein's unbiased risk estimator, which are standard procedures in low-rank matrix denoising \citep{Candes2012}. By comparing Figures \ref{FigCheckerboard}(a) and \ref{FigCheckerboard}(c) we observe that the MAP estimate is indeed very close the the original image $\bx_0$, confirming that the nuclear norm prior is a good model for low-rank signals (the estimation mean-squared error is $6.45 \times 10^{-4}$, which is $15$ times better than the original error of $0.01$). Note however that this prior is a simplistic model for $\bx_0$ in the sense that it does not include many of its main features; e.g., that $\bx_0$ is piecewise constant, periodic, highly symmetric, or that its pixel only take 3 values. Also, its representation of the singular values is approximate given that the true singular values are perfectly sparse rather than exponentially distributed. Therefore it is interesting to examine if the predictions of the model exhibit all the relevant features of $\by$, or if they highlight limitations of \eqref{Xposterior}. 

Figures \ref{FigCheckerboard}(d)-(i) depict six random replicas of $\by$ drawn from the posterior predictive distribution $f(\by^{rep}|\by)$ generated with P-MALA. We observe that the replicas are visually very similar to the original observation depicted in Figure \ref{FigCheckerboard}(b) and exhibit all of the main structural features of the checkerboard pattern that we mentioned above (e.g., periodicity, symmetries, etc.) as well as a grey-level histogram that is very similar to that of $\by$. This suggests that the model is indeed capturing the main visual characteristics our data. The replicas for this experiment were generated by using P-MALA to simulate $N=20\,000$ samples $\{X^{(t)}, t=1,\ldots,N\}$ distributed according to \eqref{Xposterior}, and then sampling $Y^{rep (t)}|X^{(t)} \sim \mathcal{N}[X^{(t)},\boldsymbol{I}\sigma^2]$ (the pictures displayed in Figure \ref{FigCheckerboard}(d)-(i) correspond to $t = 7\,500, 10\,000, 12\,500, 15\,000, 17\,500$, and $20\,000$). To implement P-MALA for \eqref{Xposterior} we used the exact proximity mapping
$$
\prox_g^{\delta/2}(\bx) = SVT[(\delta \by + 2\sigma^2 \bx)/(\delta + 2\sigma2), \alpha\delta\sigma^2/(\delta + 2\sigma2)],
$$ 
where $SVT(\bx,\tau)$ denotes the singular value soft-thresholding  operator on $\bx$ with threshold $\tau$, that is evaluated by computing the singular value representation of $\bx$ and replacing the singular values $\{s_i: i = 1,\ldots,\min(n1,n2)\}$ with $\max{(s_i - \tau,0)}$. We used $2\,000$ burn-in iterations, a thinning factor of $100$ to reduce the algorithm's memory foot-print, and tuned the value of $\delta$ to achieve an acceptance probability of approximately $50\%$. 
%We stress at this point that efficient simulation from \eqref{Xposterior} is challenging given that $X$ has dimension $n\times m = 4\,096$ and that the term $||X||_*$ is not continuously differentiable.

To illustrate the good mixing properties of P-MALA for this $4\,096$-dimensional simulation problem, Figure \ref{FigCheckerboard2} shows a $1\,000$-sample trace plot and an autocorrelation function plot of the chain $\{X^{(t)}, t=1,\ldots,N\}$, where we have used $g[X^{(t)}]$ as scalar summary. The computing time, ESS and time-normalised ESS for this experiment are $19$ minutes, $7\,930$ samples and $7.05$ samples per second. For comparison, repeating this experiment with a random walk  (RWMH) algorithm required $6.5$ minutes and produced a time-normalised ESS of $0.23$ samples per second, approximately $30$ times worse than P-MALA. Finally, note that MALA and HMC are not well defined for this model because $||\bx||_*$ is not differentiable at points where $\bx$ is rank deficient. From a practical standpoint one can still apply MALA to \eqref{Xposterior} because the probability of reaching a non-differentiable state is zero, however in our experience MALA does require $\pi \in \mathcal{C}^1$ to perform well. Repeating this experiment with MALA produced a time-normalised ESS of $0.08$ samples per second, $90$ times worse than P-MALA and $30$ times worse than RWHM (results computed by setting $\delta$ to achieve an acceptance rate of approximately $60\%$ and by computing the gradient of $||\bx||_*$ via singular-value decomposition \citep{Papadopoulo}).%, the total computing time for MALA was $20$ minutes).

\begin{figure}[htbp!]
\begin{minipage}[l2]{.33\linewidth}
  \centering
  \centerline{\includegraphics[width=6.5cm]{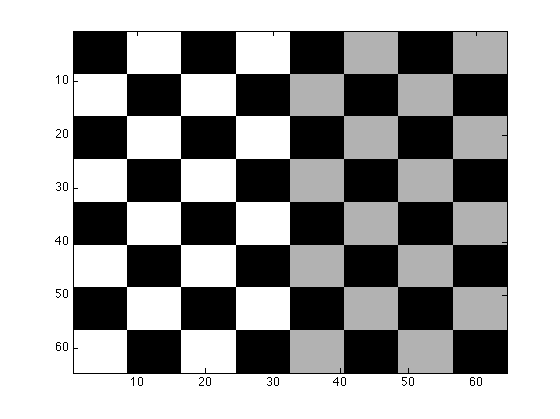}}
  \small{(a)}
\end{minipage}
\begin{minipage}[l2]{.33\linewidth}
  \centering
  \centerline{\includegraphics[width=6.5cm]{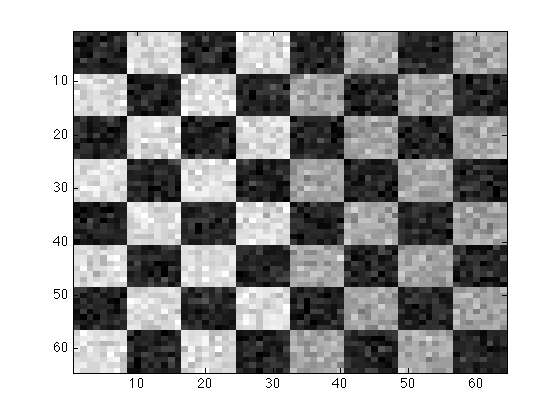}}
  \small{(b)}
\end{minipage}
\begin{minipage}[l2]{.33\linewidth}
  \centering
  %\centerline{\includegraphics[width=6.5cm]{matrixPrediction1.png}}
  \centerline{\includegraphics[width=6.5cm]{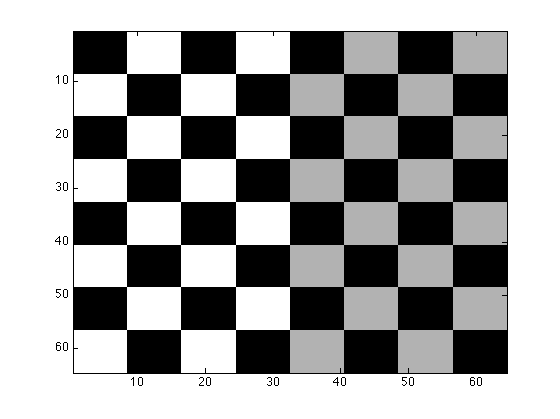}}
  \small{(c)}
\end{minipage}
\begin{minipage}[l2]{.33\linewidth}
  \centering
  \centerline{\includegraphics[width=6.5cm]{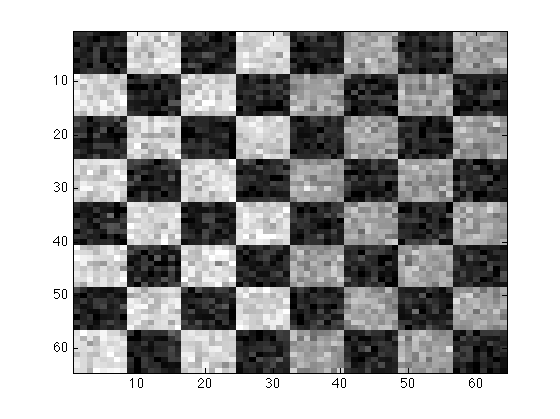}}
  \small{(d)}
\end{minipage}
\begin{minipage}[l2]{.33\linewidth}
  \centering
  \centerline{\includegraphics[width=6.5cm]{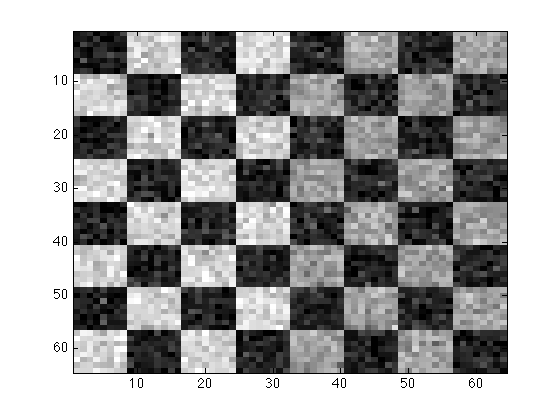}}
  \small{(d)}
\end{minipage}
\begin{minipage}[l2]{.33\linewidth}
  \centering
  \centerline{\includegraphics[width=6.5cm]{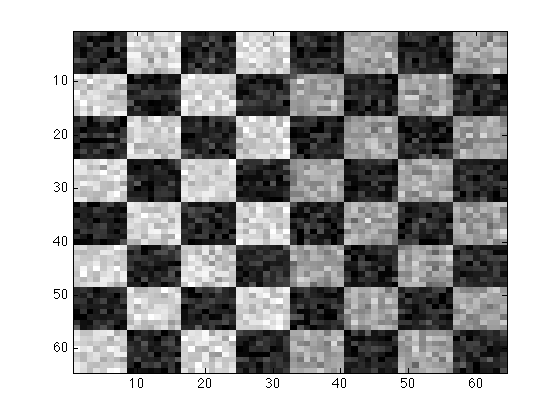}}
  \small{(f)}
\end{minipage}
\begin{minipage}[l2]{.33\linewidth}
  \centering
  \centerline{\includegraphics[width=6.5cm]{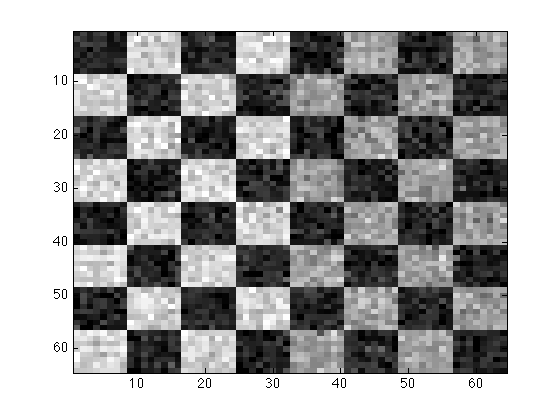}}
  \small{(g)}
\end{minipage}
\begin{minipage}[l2]{.33\linewidth}
  \centering
  \centerline{\includegraphics[width=6.5cm]{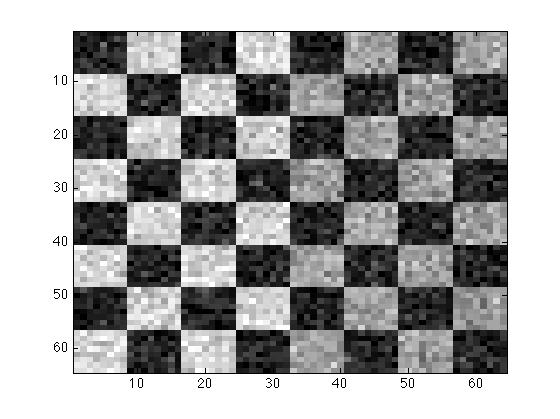}}
  \small{(h)}
\end{minipage}
\begin{minipage}[l2]{.33\linewidth}
  \centering
  \centerline{\includegraphics[width=6.5cm]{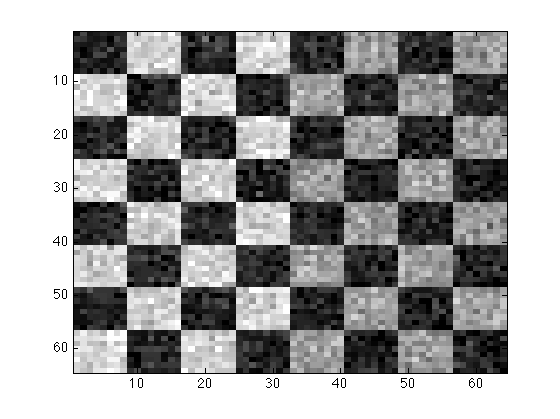}}
  \small{(i)}
\end{minipage}
\caption{\small{(a) Original checkerboard image $\bx_0$ ($64\times64$ pixels, rank $2$), (b) Noisy observation $\by = \bx_0 + \bw$, (c) MAP estimate associated with \eqref{Xposterior}, (d)-(i)  Six replicas of $\by$ generated by sampling from the posterior predictive distribution $f(\by^{rep}|\by)$.}} \label{FigCheckerboard}
\end{figure}

\begin{figure}[htbp!]
\begin{minipage}[l2]{.5\linewidth}
  \centering
  \centerline{\includegraphics[width=9.5cm]{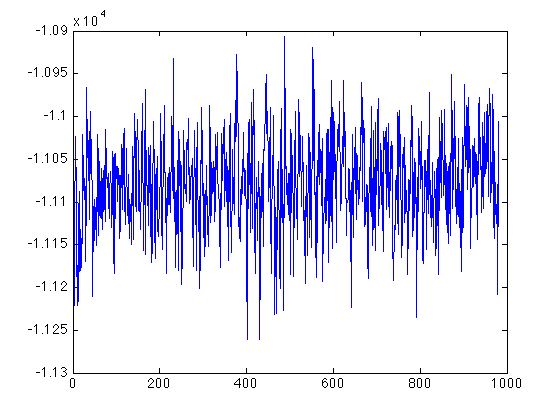}}
  \small{(a)}
\end{minipage}
\begin{minipage}[l2]{.5\linewidth}
  \centering
  \centerline{\includegraphics[width=9.5cm]{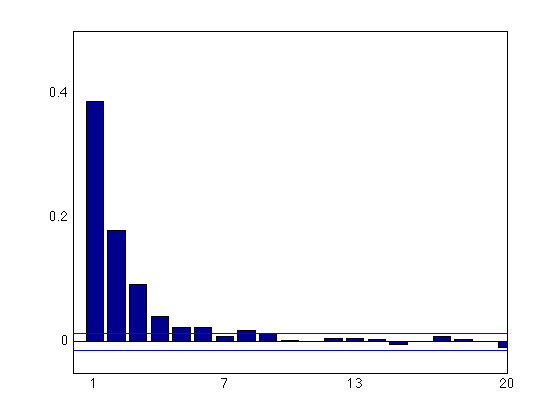}}
  \small{(b)}
\end{minipage}
\caption{\small{(a) 1000-sample trace plot and (b) autocorrelation plot using $g[X^{(t)}]$ as scalar summary.}} \label{FigCheckerboard2}
\end{figure}

\section{Conclusion}\label{sec:conclusion}
This paper studied a new Langevin MCMC algorithm that use convex analysis, namely Moreau approximations and proximity mappings, to sample efficiently from high-dimensional densities $\pi$ that are log-concave and possibly not continuously differentiable. This method is based on a new first-order approximation for Langevin diffusions that is constructed by first approximating the original diffusion $Y(t)$ with an auxiliary Langevin diffusion $Y_\lambda(t)$ with ergodic measure $\pi_\lambda$, and then discretising $Y_\lambda(t)$ using a forward Euler scheme with time step $\delta=2\lambda$. The resulting Markov chain, P-ULA, is similar to ULA except for the fact that it uses proximity mappings of $\log\pi$ instead of gradient mappings. This modification leads to a chain with favourable convergence properties that is geometrically ergodic in many cases for which ULA is transient or explosive. The proposed sampling method, P-MALA, combines P-ULA with a  step guaranteeing convergence to the desired target density. It is shown that P-MALA inherits the favourable convergence properties of P-ULA and is geometrically ergodic in many cases for which MALA does not converge geometrically and for which manifold MALA is only geometric if the time step is sufficiently small. Moreover, because P-MALA uses proximity mappings instead of gradients it can be applied to target densities that are not continuously differentiable, whereas MALA and manifold MALA require $\pi \in \mathcal{C}^1$ and $\pi \in \mathcal{C}^2$ to perform well. Finally, P-MALA was validated and compared to other MCMC algorithms through illustrative examples and applications to real data, including two challenging high-dimensional experiments related to image deconvolution and low-rank matrix denoising. These experiments show that P-MALA can make Bayesian inference techniques practically feasible for high-dimensional and non-differentiable models that are not well addressed by the existing MCMC methodology. 

Moreover, although only directly applicable to log-concave distributions, P-MALA can be used within a Gibbs sampler to simulate from more complex models. For example, it can be easily applied to a large class of bilinear models of the form \eqref{deconvolution} in which there is uncertainty about the linear operator $H$ (e.g., semi-blind image restoration), as this models can be conveniently split into two high-dimensional conditional densities that are log-concave. Similarly, its application to hierarchical models involving unknown regularisation or noise power hyper-parameters is also straightforward. Future works will focus on the application of P-MALA to the development of new statistical signal and image processing methodologies. In particular, we plan to develop a general set of tools for computing Bayesian estimators and credibility regions for high-dimensional convex linear and bilinear inverse problems, as well as stochastic optimisation algorithms for empirical Bayes estimation in signal and image processing. Another important perspective for future work is to investigate the rate of convergence of P-MALA as a function of the dimension of $\bx$. This cannot be achieved with the mathematical techniques used in of Theorems \ref{Theo1} \ref{Theo5} and \ref{Theo6}, and will require using a more appropriate set of techniques based on the Wasserstain framework (see \citet{Ottobre2014} for more details). Preliminary analyses suggest that P-MALA's mixing time depends on the shape of (the tail of) $\pi$, unlike the random walk  algorithm and MALA whose scaling is, under some conditions, independent of $\pi$.

Also, in some applications the performance of P-MALA could be improved by introducing some form of adaptation or preconditioning that captures the local geometry of the target density. This could be achieved by learning the density's covariance structure online \citep{Atachade2006} or by using an appropriate position-dependent metric. For models with $\pi \in \mathcal{C}^2$ this metric can be derived from the Fisher information matrix or the Hessian matrix as suggested in \citet{Girolami2011}, and for other log-concave densities perhaps by using preconditioning techniques from the convex optimisation literature, such as \citet{Marnissi} for instance. A key factor will be the availability of efficient algorithms for evaluating proximity mappings on non-canonical Euclidean spaces (i.e., defined using a quadratic penalty functions of the form $(\bu-\bx)^TA(\bx)(\bu-\bx)$ for some positive definite matrix $A(\bx)$). This topic currently receives a lot of attention in the optimisation literature and is the focus of important engineering efforts. Alternatively, one could also consider extending our methods to other diffusions that are more robust to anisotropic target densities \citep{Stramer99a,Stramer99b,Stramer2002}. 

We emphasise at this point that P-MALA complements rather than substitutes existing MALA and HMC methods by making high-dimensional simulation more efficient for target densities that are log-concave and have fast proximity mappings, in particular when they are not continuously differentiable. However, there are many models for which state-of-the-art MALA and HMC methods perform very well and for which P-MALA would not be applicable or computationally competitive. 

Finally, we acknowledge that since the first preprint of this work \citep{Pereyra2013}, two other works have independently proposed using proximity mappings in MCMC algorithms. These algorithms are similar to P-MALA in that they use thresholding operators within MALA and HMC algorithms (thresholding operators are a particular type of proximity mapping), but otherwise differ significantly from P-MALA. In particular, \citet{Schreck2013} considers a MALA for a variable selection problem and uses thresholding/shrinking operators to design a proposal distribution with atoms at zero (i.e., that generates sparse vectors with positive probability). \citet{Chaari2014} also considers an algorithm for a similar variable selection problem related to signal processing. Similarly to \citet{Schreck2013} that algorithm also uses thresholding operators, but to approximate gradients within an HMC leap-frog integrator. However, because thresholding operators are not continuously differentiable it is not clear if this integrator preserves volume and more crucially if the resulting HMC algorithm converges exactly to the desired target density.

\section*{Acknowledgments}
The author would like to thank the editor and two anonymous reviewers for their valuable suggestions to improve the manuscript. The author is also grateful to Ioannis Papastathopoulos, Gersende Fort, Nick Whiteley, Peter Green, Jonathan Rougier, Guy Nason, Nicolas Dobigeon, Steve McLaughlin, Hadj Batatia and Jean-Christophe Pesquet for helpful comments. Marcelo Pereyra currently holds a Marie Curie Intra-European Fellowship for Career Development. This work was in part supported by the SuSTaIN program - EPSRC grant EP/D063485/1 - at the Department of Mathematics, University of Bristol, and by a French Ministry of Defence postdoctoral fellowship.

\bibliography{bibliography}
\bibliographystyle{agsm}

\end{document}